\documentclass[11pt,a4paper]{article}
\usepackage{amsmath,amssymb,amsthm}
\newcommand{\href}[2]{#2}
\newcommand{\nolinkurl}[1]{\texttt{#1}}
\ifdefined
\fi
\ifdefined\pdfinfoomitdate\pdfinfoomitdate=1\fi
\ifdefined\pdftrailerid\pdftrailerid{}\fi
\ifdefined\pdfsuppressptexinfo\pdfsuppressptexinfo=-1\fi
\newtheorem{theorem}{Theorem}
\newtheorem{lemma}[theorem]{Lemma}

\newtheorem{corollary}[theorem]{Corollary}
\theoremstyle{remark}

\newtheorem*{unnumberedremark}{Remark}
\newcommand{\Ed}{E_F^{\,\mathrm d}}
\newcommand{\Ec}{E_F^{\,\mathrm c}}
\newcommand{\J}{\mathcal J}
\newcommand{\E}{\mathcal E}
\newcommand{\U}{\mathcal U}
\newcommand\HH{\mathcal H}
\newcommand\St{\mathfrak S}
\newcommand{\ext}{\mathrm{ext}}
\newcommand{\tr}{\operatorname{Tr}}

\newcommand{\ket}[1]{\lvert #1\rangle}
\newcommand{\bra}[1]{\langle #1\rvert}
\newcommand{\proj}[1]{\ket{#1}\bra{#1}}
\newcommand{\HS}{\mathcal S_2}

\title{On the coincidence of discrete and continuous definitions of Entanglement of Formation}
\author{%
Haobo Ma\textsuperscript{2,4}\quad
Maksim E. Shirokov\textsuperscript{1}\quad
Wenlin Zhang\textsuperscript{3,4}\\[0.6em]
\small\textsuperscript{1}Steklov Mathematical Institute, Moscow, Russia\\
\small\textsuperscript{2}ChronoAI Pte Ltd, Singapore\\
\small\textsuperscript{3}National University of Singapore\\
\small\textsuperscript{4}The Omega Institute\\[0.4em]
\small\href{mailto:auric@aelf.io}{\nolinkurl{auric@aelf.io}}\quad
\href{mailto:msh@mi-ras.ru}{\nolinkurl{msh@mi-ras.ru}}\quad
\href{mailto:e1327962@u.nus.edu}{\nolinkurl{e1327962@u.nus.edu}}}
\date{15 October 2026}

\begin{document}
\maketitle
\begin{abstract}
It is proved that every Borel ensemble of  pure states on a separable bipartite Hilbert
space with finite average entanglement admits a countable
ensemble with the same barycenter and arbitrarily small additional
average entanglement. Consequently, the discrete and continuous versions
of the Entanglement of Formation coincide on the whole space of states of infinite-dimensional bipartite quantum system (with infinite
values allowed).

The main idea of the proof (realized with the AI-assistance) is the lifting a Borel ensemble of pure bipartite states to a Borel probability measure on the unit sphere in the space of Hilbert--Schmidt operators followed by dividing the support of this measure into appropriate cells and replacing  each cell by its conditional barycenter. The resulting countable ensemble is obtained by a recursive construction.
\end{abstract} 

\section{Introduction}

The Entanglement of Formation (EoF) is one of the fundamental
quantitative measures of bipartite quantum entanglement. In
finite-dimensional quantum systems, it is defined as the convex-roof
extension of the entropy of entanglement of pure states \cite{Bennett}.  More
precisely, the EoF of a state $\rho$ of a bipartite system $AB$  is the infimum of
the average entanglement entropy over all pure-state decompositions
of $\rho$. This definition has become a standard tool in quantum
information theory, owing to its operational significance and its
important structural properties \cite{4H,P&V,Uhl}.

The situation becomes more subtle when the underlying Hilbert spaces
are infinite-dimensional. In this setting, the von Neumann entropy is
in general unbounded and may take the value $+\infty$, while the set
of quantum states is noncompact. Consequently, some of the familiar
relations between convexity, continuity, and convex-roof constructions
require additional care. In particular, there are two natural ways
to extend the finite-dimensional definition of the EoF to
infinite-dimensional systems: one may optimize over countable
pure-state decompositions, or over arbitrary Borel probability
measures on the set of pure states.

Let $\mathcal{H}_A$ and $\mathcal{H}_B$ be separable Hilbert spaces,
and let $\rho$ be a state on
$\mathcal{H}_{AB}\doteq\mathcal{H}_A\otimes\mathcal{H}_B$. The discrete version of the EoF
is defined by
\begin{equation}
E_F^{\,\rm d}(\rho)
=
\inf_{\sum_{i}
q_i |\psi_i\rangle\langle\psi_i|=\rho}
\sum_{i=1}^{+\infty} q_i
S\!\left(\operatorname{Tr}_B
|\psi_i\rangle\langle\psi_i|\right),
\label{eq:discrete-eof}
\end{equation}
where the infimum is taken over all countable pure-state
decompositions of $\rho$, and $S(\cdot)$ denotes the
von Neumann entropy. The continuous version is given by
\begin{equation}
E_F^{\,\rm c}(\rho)
=
\inf_{\mu:\,\int |\psi\rangle\langle\psi|\,\mu(d\psi)=\rho}
\int
S\!\left(\operatorname{Tr}_B
|\psi\rangle\langle\psi|\right)\,\mu(d\psi),
\label{eq:continuous-eof}
\end{equation}
where the infimum is taken over all Borel probability measures
$\mu$ on the set $\ext\St(\HH_{AB})$ of pure states with barycenter $\rho$  (here and in what follows we write $d\psi$ instead of $d|\psi\rangle\langle\psi|$ for brevity).
Every countable ensemble determines an atomic probability measure.
It follows immediately that
\begin{equation}
E_F^{\,\rm c}(\rho)\leq E_F^{\,\rm d}(\rho)
\label{eq:eof-order}
\end{equation}
for every bipartite state $\rho$. The converse inequality, however, does
not follow directly from the definitions.

The distinction between these two constructions is not merely
formal. The continuous convex-roof construction yields an
entanglement measure with strong analytical properties, including
lower semicontinuity and monotonicity under generalized
local operations and classical communication (LOCC) \cite{shirokov2010,cost2026}. By contrast,
in infinite dimensions a discrete convex-roof construction need
not inherit all the expected properties of the corresponding
continuous one.\footnote{One can say more: the validity of these properties for the discrete convex-roof construction implies  its coincidence with the continuous one.}\footnote{One of the main problems with the EoF  defined in (\ref{eq:discrete-eof}) was (up to the appearance of this article) the lack of the proof that it vanishes on the set of  countably-non-decomposable separable states  with infinite  marginal entropies \cite{hsw2005,shirokov2010}.}  At the same time, an obvious drawback of the definition (\ref{eq:continuous-eof}) is its technical complexity: the EoF is a hardly computable function even in the
finite-dimensional case, but the necessity to take  the infimum in (\ref{eq:continuous-eof}) over all measures on $\ext\St(\HH_{AB})$ with fixed
barycenter radically complicates the definition. Moreover, from the physical point of view, we are faced with the following problem: if $\rho$ is a state with finite
energy (w.r.t. some Hamiltonian), i.e. a  physically realizable state, the supports of some probability measures on $\ext\St(\HH_{AB})$ with the barycenter $\rho$ contain states with infinite
energy which can not be realized in physical experiments.\footnote{For a Hamiltonian bounded below, if the average state of a discrete ensemble has finite energy then every state with positive weight has finite energy. For a general Borel probability measure with a barycenter of finite energy, the states have finite energy almost everywhere.} This issue was studied in \cite[Section~5.1]{shirokov2010}, where
it is proved, in particular, that the two definitions of the EoF coincide
at any state having at least one marginal (reduced)
state with finite von Neumann entropy, but the case of an arbitrary
bipartite state remained unresolved.

The aim  of the present paper is to settle this question by
proving that the two definitions coincide for every bipartite state
on a tensor product of separable Hilbert spaces:
\begin{equation}
E_F^{\,\rm d}(\rho)=E_F^{\,\rm c}(\rho)
\quad\text{for every bipartite state }\rho.
\label{eq:main-result}
\end{equation}
In particular, the equality holds without requiring either
marginal state to have finite entropy. Thus the discrete
definition, based exclusively on countable pure-state ensembles,
and the continuous definition, based on arbitrary probability
measures, determine the same value even in the most general
infinite-dimensional setting considered here.

This result has several consequences. First, it removes a
potential ambiguity in the definition of the EoF for
infinite-dimensional systems: the numerical value of the measure
does not depend on whether one allows only countable ensembles
or arbitrary continuous ensembles. Second, since the continuous
EoF is the lower-semicontinuous convex-roof extension of the
pure-state entropy of entanglement, the equality
\eqref{eq:main-result} establishes the lower semicontinuity of
the discrete EoF as well. Third, it allows the analytical
properties proved for the continuous construction to be applied
directly to the discrete definition, including the corresponding
generalized LOCC-monotonicity properties. These conclusions are
particularly useful when working with limits of states,
infinite-dimensional decompositions, and approximation
procedures in quantum information theory.

It is important to stress that equality of the infima in (\ref{eq:discrete-eof}) and (\ref{eq:continuous-eof}) does
not imply, by itself, that the attainability of the infimum  in (\ref{eq:continuous-eof}) proved in \cite{shirokov2010}
guarantees the attainability of the infimum  in (\ref{eq:discrete-eof}). Rather, the result
states that allowing arbitrary probability measures cannot lower
the infimum below the value obtainable as an infimum over
countable ensembles.

\section{The main result}\label{sec:main}

In this section, we prove that  definitions (\ref{eq:discrete-eof}) and (\ref{eq:continuous-eof}) coincide. The main result implying this coincidence is an exact discretization
theorem: a continuous pure-state ensemble of finite average
entanglement can be replaced by a countable pure-state ensemble with
the same barycenter and arbitrarily small increase in entanglement cost.\footnote{In the article,  the term entanglement cost (or, briefly, cost) of an ensemble (measure) means the average entanglement of this ensemble (measure). This should not be confused with the entanglement cost of a bipartite state in terms of \cite{cost2026}.}

Let $\mathcal{H}_A,\mathcal{H}_B$ be separable complex Hilbert spaces, and let
$\mathcal{H}_{AB}=\mathcal{H}_A\otimes \mathcal{H}_B$. All logarithms are natural. For a unit vector
$\psi\in \mathcal{H}_{AB}$ put
\[
 e(\psi)=S(\tr_B\proj{\psi})\in[0,\infty],
 \qquad S(\sigma)=-\tr\sigma\log\sigma.
\]
The function $e$ depends only on the pure state. Using this function we can rewrite definitions (\ref{eq:discrete-eof}) and (\ref{eq:continuous-eof}) as
\begin{align}
 \Ed(\rho)&=\inf\left\{\sum_jq_je(\psi_j):
       q_j>0,\ \sum_jq_j=1,\ \sum_jq_j\proj{\psi_j}=\rho\right\},
       \label{eq:discrete}\\
 \Ec(\rho)&=\inf_{\mu:\,\int\pi\,\mu(d\pi)=\rho}
       \int e(\pi)\,\mu(d\pi).
       \label{eq:continuous}
\end{align}
In \eqref{eq:discrete} the ensemble is finite or countable and its
operator sum converges in trace norm. In \eqref{eq:continuous} $\mu$
is a Borel probability measure on pure states, with the trace norm
topology, and the barycenter is defined via a Bochner integral.\footnote{It is reasonable to consider such measures as generalized ensembles of quantum states \cite{holevo2005,H-SCI}.}

The approach developed in \cite{shirokov2004,shirokov2010}  establishes the equality $\Ed(\rho)=\Ec(\rho)$ when at least
one marginal entropy of the state $\rho$  is finite and proves lower semicontinuity of the
function $\Ec$.

Lemma 1 in \cite{holevo2005} says that, on the full state space, measures
with finite support are dense among all probability measures having a
given barycenter, in the topology of weak convergence of measures.
Its proof replaces each cell by its conditional state barycenter,
which can be mixed. The theorem below gives pure output atoms and
an upper bound on average entanglement at the exact barycenter.

\begin{theorem}\label{thm:discretization}
Let $\,\mu$ be a Borel probability measure on the set $\,\ext\St(\HH_{AB})$ of pure states on $\mathcal{H}_{AB}$ with
barycenter $\rho$ and finite average entanglement
$c_\mu=\int e(\pi)\,\mu(d\pi)$. For every $0<\eta<1$ there is a countable pure-state ensemble with the average state $\rho$ with average entanglement
at most
\begin{equation}\label{eq:finalbound}
 c_\mu+\frac{5\eta}{4-\eta}.
\end{equation}
\end{theorem}

\begin{unnumberedremark}
For every $\varepsilon>0$,
Theorem~\ref{thm:discretization} gives a countable ensemble with
cost at most $c_\mu+\varepsilon$ by choosing $\eta$ sufficiently
small. A countable ensemble with cost at most $c_\mu$ need not
exist: zero cost is not attained for a countably nondecomposable
separable state, as explained in Corollary~\ref{cor:consequences}.
\end{unnumberedremark}

Here discretization permits new pure states, including states of
infinite Schmidt rank; the output atoms need not belong to the support of the
original measure $\mu$.

\begin{corollary}\label{cor:coincidence}
The equality  $\Ed(\rho)=\Ec(\rho)$ holds for every state $\rho$ on $\mathcal{H}_{AB}$,
with infinite values allowed.
\end{corollary}

\begin{proof} Every countable ensemble defines a Borel probability measure on
pure states with the same barycenter and the same cost.
Hence $\Ec(\rho)\le\Ed(\rho)$ in the extended nonnegative reals.
If $\Ec(\rho)<+\infty$, fix $\varepsilon>0$ and choose a Borel
ensemble with the entanglement cost equal to $\Ec(\rho)$. The existence of such an ensemble is proved in \cite[Section 5]{shirokov2010}.
Choose $0<\eta<1$ so that $5\eta/(4-\eta)<\varepsilon$.
By Theorem~\ref{thm:discretization}, there is a countable ensemble
with the cost less than $\Ec(\rho)+\varepsilon$, and therefore
$\Ed(\rho)\le\Ec(\rho)+\varepsilon$.
Letting $\varepsilon\downarrow0$ gives the reverse inequality.
If $\Ec(\rho)=+\infty$, the already established inequality
$\Ec(\rho)\le\Ed(\rho)$ forces $\Ed(\rho)=+\infty$.
\end{proof}

\section{The proof of Theorem \ref{thm:discretization}}\label{sec:proof}

\subsection{Vectors and amplitude operators}\label{sec:amplitudes}
Fix orthonormal bases $(\phi_k)$ of $\mathcal{H}_A$ and $(\varphi_l)$ of $\mathcal{H}_B$.
Write a vector $\psi$ in $\HH_{AB}$ and its coefficient operator as
\begin{equation}\label{eq:coefficients}
 |\psi\rangle=\sum_{k,l}a_{kl}\,|\phi_k\otimes \varphi_l\rangle,
 \qquad A |\varphi_l\rangle=\sum_k a_{kl}|\phi_k\rangle.
\end{equation}
The condition $\sum_{k,l}|a_{kl}|^2<\infty$ is exactly the condition
that $A:\mathcal{H}_B\to \mathcal{H}_A$ is a Hilbert--Schmidt operator, and
\[
 \|\psi\|^2=\sum_{k,l}|a_{kl}|^2=\tr A^*A=\|A\|_2^2.
\]
Thus \eqref{eq:coefficients} is a linear isometric identification
of $\mathcal{H}_{AB}$ with $\HS(\mathcal{H}_B,\mathcal{H}_A)$ -- the separable Hilbert space of all Hilbert--Schmidt operators from $\mathcal{H}_B$ to $\mathcal{H}_A$ with the inner product
$$
\langle A,B\rangle\doteq \tr A^*B.
$$

We call operators in $\HS(\mathcal{H}_B,\mathcal{H}_A)$ \emph{amplitude operators} or, briefly, \emph{amplitudes}.  Taking the partial trace in the above bases gives
\[
 \langle \phi_k,(\tr_B\proj{\psi})\phi_{k'}\rangle
       =\sum_l a_{kl}\overline{a_{k'l}}
       =\langle \phi_k,AA^*\phi_{k'}\rangle.
\]
For infinite sums, first perform this computation for finite
coefficient matrices and pass to the limit. Indeed,
\[
 \|AA^*-BB^*\|_1
 \le(\|A\|_2+\|B\|_2)\|A-B\|_2,
\]
and the analogous rank-one estimate holds for the bipartite
vectors; partial trace is continuous in trace norm. Consequently
$\tr_B\proj{\psi}=AA^*$ for every $\psi\in \mathcal{H}_{AB}$.

If $A=\sum_j s_j\ket{u_j}\bra{v_j}$ is a singular value
decomposition of $A\in\HS(\mathcal{H}_B,\mathcal{H}_A)$, the corresponding vector is
$|A\rangle\doteq\sum_j s_j |u_j\otimes\overline{v_j}\rangle$, where the bar is coordinate
conjugation in the basis $(\varphi_l)$. This is a Schmidt decomposition of the vector $|A\rangle$. In
particular, the nonzero eigenvalues of $AA^*$ are $s_j(A)^2$.
We list singular values in decreasing order and append zeros when
necessary. For a unit amplitude $A$ this gives
\[
 e(A)=-\sum_j s_j(A)^2\log s_j(A)^2.
\]
We use $e(A)$ as shorthand for the entanglement of its associated
unit vector. For arbitrary amplitudes define the homogeneous cost
\begin{equation}\label{eq:homogeneous}
 \E(A)=\|A\|_2^2 S\left(\frac{AA^*}{\|A\|_2^2}\right)
 \quad(A\ne0),\qquad \E(0)=0.
\end{equation}
For $z\ne0$, normalization removes $|z|^2$, so
$\E(zA)=|z|^2\E(A)$; at zero we use the displayed convention.
Thus a weighted pure atom with amplitude $y=\sqrt q\,A$ and
$\|A\|_2=1$ contributes exactly $\E(y)=q e(A)$.
The notation $\proj{A}$ always denotes the rank-one operator on the
\emph{bipartite} space $\mathcal{H}_{AB}$. It is distinct from $AA^*$, an operator
on $\mathcal{H}_A$. Both have trace $\|A\|_2^2$.

\subsection{Overview of the proof}\label{sec:overview}
The proof separates one extraction step from the iteration of that
step. The objects being averaged are unit amplitude operators, not
pure states  (in contrast to the proof of Lemma 1 in \cite{holevo2005}!). This distinction allows the average to remain a
vector and hence to determine a pure output atom.

First choose a Borel unit-vector representative of each input pure state. We introduce a convex lower semicontinuous spectral
functional $\J$ on amplitude space that equals $3+e$ on unit amplitudes
and satisfies
\[
 \J(D)\ge2\|D\|_2^2,
 \qquad \E(D)\le\J(D)\quad\text{if }\|D\|_2\le1.
\]
The unit amplitudes of finite entanglement admit a countable Borel
partition into cells such that $\J(A-B)<\eta$ for every pair in a
cell. The partition uses approximation by finite-rank operators in
$\J$, followed by approximation of their coefficients by rational
numbers. Hilbert--Schmidt norm density alone would not give the
needed entropy estimate.

For a cell $C$ of measure $m>0$, let $M=m^{-1}\int_C A\,\mu(dA)$.
Replace part of its state barycenter by the pure atom with
unnormalized amplitude $y=\sqrt m M$. The identity
\[
 \int_C\proj{A}\,\mu(dA)-\proj{y}
   =\frac1{2m}\int_{C\times C}\proj{A-B}\,\mu(dA)\,\mu(dB)
\]
shows that the part left over is positive. It also represents that
part by the differences $D=A-B$. Relative to the pair measure
$(2m)^{-1}(\mu|_C)\otimes(\mu|_C)$, each nonzero difference
contributes the normalized vector $D/\|D\|_2$ with density
$\|D\|_2^2$.
This weight is essential: the entropy of the normalized difference
can be large, but its weighted cost is $\E(D)<\eta$.

For an incoming measure of mass $t$ and cost $c$, summing over the
cells gives a residual measure of mass $t'\le\eta t/4$ and cost
$c'\le\eta t/2$. Jensen's inequality for $\J$ bounds the cost of
the extracted atoms by $c+3t'$. The residual measure is again a
finite Borel measure on unit amplitudes of finite entanglement, so
exactly the same construction applies to it.

Starting with $t_0=1$ and $c_0=c_\mu$, repeat with a fixed
$0<\eta<1$. The residual trace is at most $(\eta/4)^n$ after
$n$ steps. Every step is an exact positive decomposition of the
previous remainder. Therefore the countable union of the extracted
atoms sums in trace norm to the original state. The residual costs
and the extraction errors form summable series, whose total is at
most $5\eta/(4-\eta)$. This proves the discretization theorem and,
by taking $\eta\downarrow0$, equality of the two roofs.

Sections~\ref{sec:functional} and~\ref{sec:cells} construct the
functional and the cells. Section~\ref{sec:extraction} proves the
extraction step, including the residual measure, and
Section~\ref{sec:iteration} carries out the iteration.

\subsection{A convex function of amplitude operators}\label{sec:functional}
Define the function
\begin{equation}\label{eq:f}
 \J(A)=\sum_{j\ge1}f(s_j(A)),\quad\textrm{ where }\quad f(t)=
 \begin{cases}
 t^2(3-2\log t),&0<t\le1,\\
 2t^2+1,&t\ge1,\\
 0,&t=0,
 \end{cases}
\end{equation}
on $\HS(\mathcal{H}_B,\mathcal{H}_A)$.\footnote{It is essential that the function $f(t)$ is  convex and continuously differentiable on $\mathbb{R}_+$ and that $\,f(0)=0\,$ and $\,f(t)\geq t^2(3-2\log t)$ for all $t>0$.} The sum may be infinite. On the Hilbert--Schmidt unit sphere,
\begin{equation}\label{eq:unit}
 \J(A)=3+e(A).
\end{equation}

\begin{lemma}\label{lem:convex}
The function $\J$ is convex and lower semicontinuous in the Hilbert--Schmidt
norm. For all $A,B$ in $\HS(\mathcal{H}_B,\mathcal{H}_A)$  and every $a\ge1$,
\begin{align}
 \J(A)&\ge2\|A\|_2^2,\label{eq:normbound}\\
 \J(aA)&\le a^2\J(A),\label{eq:scaling}\\
 \J(A+B)&\le2\J(A)+2\J(B).\label{eq:quasitriangle}
\end{align}
If $\,d=\|D\|_2^2\le1$, then, with $\,0\log0=0$,
\begin{equation}\label{eq:entropyidentity}
 \E(D)=\J(D)-3d+d\log d\le\J(D).
\end{equation}
\end{lemma}
\begin{proof}
For $0<t<1$,
$f'(t)=4t(1-\log t)$ and $f''(t)=-4\log t\ge0$.
For $t>1$, $f'(t)=4t$ and $f''(t)=4$. Both $f$ and $f'$ match
at $1$, and $f'(0+)=0$. Hence $f$ is increasing and convex.
Moreover
\[
 \frac{f(t)}{t^2}=
 \begin{cases}3-2\log t,&0<t\le1,\\2+t^{-2},&t\ge1\end{cases}
\]
is decreasing and is at least $2$: its derivatives on the two
pieces are $-2/t$ and $-2/t^3$, and its values agree at $1$.
It follows term by term that $f(s_j(A))\ge2s_j(A)^2$; summing gives
\eqref{eq:normbound}. For $a\ge1$ and $t>0$, monotonicity of this
ratio gives
\[
 \frac{f(at)}{a^2t^2}\le\frac{f(t)}{t^2},
 \qquad f(at)\le a^2 f(t).
\]
The latter also holds at $t=0$. Since $s_j(aA)=a s_j(A)$, summing
proves \eqref{eq:scaling}, including when $\J(A)=+\infty$.
The second branch of $f$ extends convexity beyond $1$ while keeping
this scaling estimate. It allows us to apply convexity on the
whole vector space, even though unit amplitudes have singular values
at most $1$.

For $u>0$ let
\[
 K_u(A)=\sum_j(s_j(A)-u)_+
       =\sup_{n\ge0}\left(\sum_{j=1}^n s_j(A)-nu\right).
\]
Each finite Ky Fan sum $\,\sum_{j=1}^n s_j(A)\,$ is a convex function of $A$. Indeed, its variational
formula is the supremum of
$\operatorname{Re}\sum_{j=1}^n\langle u_j,Av_j\rangle$
over orthonormal systems $(u_j)$ and $(v_j)$, and so is a supremum of
real linear functionals of $A$. If a space has finite dimension,
restrict $n$ to the available dimension. Thus $K_u$ is convex.
Only finitely many singular values exceed $u$, so $K_u(A)<+\infty$
for a Hilbert--Schmidt operator.
Since $f'(0+)=f(0)=0$ and $f'$ is locally absolutely continuous,
\[
 f(t)=\int_0^{+\infty}(t-u)_+f''(u)\,du.
\]
Tonelli's theorem now gives
\[
 \J(A)=\int_0^{+\infty} K_u(A)f''(u)\,du,
\]
which proves convexity also for extended values. The singular values
are individually continuous in Hilbert--Schmidt norm, and $f\ge0$.
The supremum of the continuous finite partial sums defining $\J$ is
therefore lower semicontinuous.
Convexity and \eqref{eq:scaling} imply
\[
 \J(A+B)\le\tfrac12\J(2A)+\tfrac12\J(2B)
          \le2\J(A)+2\J(B).
\]
Finally, if $0<d\le1$, put $x_j=s_j(D)^2$. Then
$0\le x_j\le1$, $\sum_jx_j=d$, and
\[
 \J(D)=3d-\sum_jx_j\log x_j,
 \qquad
 \E(D)=-\sum_j x_j\log(x_j/d)
       =-\sum_jx_j\log x_j+d\log d.
\]
Subtracting gives \eqref{eq:entropyidentity}. The entropy sums are
nonnegative, and the other terms are finite, so the identity also
holds if either entropy sum is infinite. At $d=0$, $D=0$ and both
sides vanish. Finally $-3d+d\log d\le0$ for $0\le d\le1$,
which proves the upper bound.
\end{proof}

\subsection{Measurable representatives and countable cells}\label{sec:cells}
Set
\[
 \U=\{A\in\HS(\mathcal{H}_B,\mathcal{H}_A):\|A\|_2=1,\ \J(A)<+\infty\}.
\]
The unit sphere is closed and
$\{\J<+\infty\}=\bigcup_{n\ge1}\{\J\le n\}$ is Borel by
the lower semicontinuity of $\J$. Thus $\U$ is a Borel subset of the separable
Hilbert space $\HS(\mathcal{H}_B,\mathcal{H}_A)$. On $\U$, $e=\J-3$ is a finite Borel function.

\subsubsection{A Borel choice of phase}\label{sec:selection}
The measure $\mu$  is a measure on pure states,
whereas taking cell means requires unit vectors. We now give the
representative explicitly. Let $\ext\St(\mathcal{H}_{AB})$ be the set of
pure states, with its trace norm Borel structure,
and fix an orthonormal basis $(h_n)$ of $\mathcal{H}_{AB}$. Put
\[
 g_n(\pi)=\langle h_n,\pi h_n\rangle,
 \qquad
 D_n=\{\pi\in\ext\St(\mathcal{H}_{AB}):g_n(\pi)>0,\ g_j(\pi)=0\text{ for }j<n\}.
\]
Each $g_n$ is nonnegative and continuous in trace norm. Hence each
$D_n$ is Borel. Since $\sum_n g_n(\pi)=\tr\pi=1$, there is a
first index with positive diagonal entry; the sets $D_n$ form a partition
of $\ext\St(\mathcal{H}_{AB})$. On $D_n$ define
\begin{equation}\label{eq:section}
 \psi(\pi)=\frac{\pi h_n}{\sqrt{g_n(\pi)}}.
\end{equation}
The map $\pi\mapsto\pi h_n$ is continuous from trace norm to vector
norm, and the denominator is positive on $D_n$. Thus this formula
is continuous on each piece and the countably defined section is
Borel. Moreover, $\pi^2=\pi$ gives
$\|\pi h_n\|^2=g_n(\pi)$, so $\|\psi(\pi)\|=1$. The vector $\psi(\pi)$
lies in the one-dimensional range of $\pi$, and therefore
$\proj{\psi(\pi)}=\pi$. Its $n$th coefficient is the positive
number $\sqrt{g_n(\pi)}$, which fixes its phase.

Composing this section with the isometry defined in \eqref{eq:coefficients} gives a Borel
map $L$ from pure states to unit amplitudes. If the original
measure is $\mu_0$, its pushforward (the image under $L$) $\widetilde\mu=\mu_0\circ L^{-1}$
satisfies
\[
 \int_{\HS(\mathcal{H}_B,\mathcal{H}_A)}\proj{A}\,\widetilde\mu(dA)=\int_{\ext\St(\mathcal{H}_{AB})}\!\!\pi\,\mu_0(d\pi),
 \quad
 \int_{\HS(\mathcal{H}_B,\mathcal{H}_A)}\!\! e(A)\,\widetilde\mu(dA)=\int_{\ext\St(\mathcal{H}_{AB})}\!\! e(\pi)\,\mu_0(d\pi).
\]
These equalities follow from $\proj{L(\pi)}=\pi$ and the
pushforward integration formula. In particular, finite average
entanglement implies $\widetilde\mu(\U)=1$. We restrict this
measure to $\U$ without changing its barycenter or cost. This is a
full-measure restriction; no assertion about its topological support
is needed. The phase choice need not be continuous across the sets
$D_n$. The partition below works for every unit amplitude in $\U$,
so this lack of global continuity causes no difficulty.

For clarity about the operator integrals, for arbitrary vectors
$x,y\in \mathcal{H}_{AB}$ one has
\begin{equation}\label{eq:rankone-continuity}
 \|\proj{x}-\proj{y}\|_1
 \le(\|x\|+\|y\|)\|x-y\|,
 \qquad \|\ket{x}\bra{y}\|_1=\|x\|\|y\|.
\end{equation}
The first inequality follows by writing the difference as
$\ket{x-y}\bra{x}+\ket{y}\bra{x-y}$. The Hilbert--Schmidt
space and the trace class on $\mathcal{H}_{AB}$ are separable. Consequently the
functions $A\mapsto \proj{A}$ and $A\mapsto A$  are strongly measurable. So, their integrals against
finite measures therefore exist as Bochner integrals. Applying the
continuous trace functional under such an integral is legitimate.

\subsubsection{A partition that controls weighted entropy}
\begin{lemma}\label{lem:cells}
For every $\eta>0$, $\U$ has a countable Borel partition $(C_i)$
such that $A,B\in C_i$ implies $\J(A-B)<\eta$.
\end{lemma}
\begin{proof}
Enumerate all coefficient matrices with finite support and rational
real and imaginary entries in the fixed bases as $(B_i)$. There are
countably many such matrices. We first prove the stronger density
statement
\begin{equation}\label{eq:dense}
 \J(A)<+\infty\quad\Longrightarrow\quad
 \inf_i\J(A-B_i)=0.
\end{equation}
Fix $\delta>0$ and choose a singular value decomposition of $A$.
Let $A_r$ retain its first $r$ terms. The singular values of
$A-A_r$ are precisely the omitted ones, including zeros, and hence
\[
 \J(A-A_r)=\sum_{j>r}f(s_j(A))\rightarrow0\quad\textrm{ as }\quad  r\rightarrow+\infty.
\]
Choose $r$ so that $\J(A-A_r)<\delta/4$.

To approximate $A_r$, write it as a sum of $r$ rank-one operators.
Approximate each of the finitely many left and right vectors by
finite-support vectors with rational real and imaginary coordinates,
and each coefficient by a rational complex number. These choices
produce finite-support rational matrices $B^{(n)}$ of rank at most
$r$ with $\,\|A_r-B^{(n)}\|_2\to0\,$  as  $\,n\to+\infty$. The convergence follows, term
by term, from
\[
 \|\ket{u}\bra{v}-\ket{u'}\bra{v'}\|_2
 \le\|u-u'\|\|v\|+\|u'\|\|v-v'\|.
\]
The rank of $A_r-B^{(n)}$ is at most $2r$. Every singular value
is bounded by its Hilbert--Schmidt norm, so
\[
 \J(A_r-B^{(n)})
 \le2r f(\|A_r-B^{(n)}\|_2)\longrightarrow0\quad\textrm{ as }\quad n\to+\infty.
\]
Choose $n$ such that $\J(A_r-B^{(n)})<\delta/4$. By
\eqref{eq:quasitriangle},
\[
 \J(A-B^{(n)})
 \le2\J(A-A_r)+2\J(A_r-B^{(n)})<\delta.
\]
This proves \eqref{eq:dense}; the case $A=0$ follows by taking
$B_i=0$. The rank restriction in the second approximation is what
permits the finite sum bound. No continuity of $\J$ on its full
effective domain has been assumed.

For the given $\eta$, the sets
\[
 U_i=\{A\in\U:\J(A-B_i)<\eta/4\}
\]
cover $\U$ by \eqref{eq:dense}. Each is Borel: the translated
functional is lower semicontinuous, and its strict sublevel set is
the union of its closed sublevel sets with rational levels less
than $\eta/4$, intersected with $\U$. Set
\[
 C_1=U_1,\qquad C_i=U_i\setminus\bigcup_{j<i}U_j\quad(i>1).
\]
These sets are Borel, disjoint and cover $\U$; empty cells can be
omitted. For $A,B\in C_i$, singular values are unchanged by a
minus sign, and therefore
\[
 \J(A-B)\le2\J(A-B_i)+2\J(B_i-B)<\eta.
\]
The centres $B_i$ label the cells; they are not output atoms, and
need not be normalized.
\end{proof}

\subsection{Exact extraction and its remainder}\label{sec:extraction}
The recursive construction uses measures of mass less than one.
We therefore formulate the extraction step for an arbitrary finite
Borel measure $\mu$ on $\U$ with finite entanglement cost. Write
\[
 R=\int_{\U}\proj{A}\,\mu(dA),\qquad
 t=\mu(\U)=\tr R,\qquad c=\int_{\U}e(A)\,\mu(dA)<+\infty.
\]
The operator $R$ is positive and trace class. The cost $c$ refers
to this particular representing measure.

\begin{lemma}\label{lem:extraction}
For every $0<\eta<1$ there are countably many vectors $y_i$ in $\HS(\mathcal{H}_B,\mathcal{H}_A)\cong\HH_{AB}$ and a
finite Borel measure $\mu'$ on $\U$ such that
\begin{align}
 R&=\sum_i\proj{y_i}+R',\qquad
       R'=\int_{\U}\proj{A}\,\mu'(dA)\ge0,\label{eq:extract}\\
 t'&:=\tr R'\le\frac{\eta}{4}t,\label{eq:remtrace}\\
 c'&:=\int_{\U}e(A)\,\mu'(dA)\le\frac{\eta}{2}t,\label{eq:remcost}\\
 \sum_i\E(y_i)&\le c+3t'.\label{eq:atomcost}
\end{align}
The sum in \eqref{eq:extract} converges in trace norm. Each atom
satisfies $\proj{y_i}\le R$. If $\,t=0$, all families and measures
may be taken to be empty.
\end{lemma}
\begin{proof}
If $t=0$, $\mu$ is the zero measure and all assertions follow.
Assume $t>0$ and use the partition from Lemma~\ref{lem:cells}.
Discard the cells of zero $\mu$-measure; their countable union is
null, so no integral changes. All sums below are over the remaining
cells.

\medskip\noindent\textit{Cell means and their costs.}
For each cell define
\begin{equation}\label{eq:cell-data}
 \begin{gathered}
 m_i=\mu(C_i)>0,\qquad
 M_i=\frac1{m_i}\int_{C_i}A\,\mu(dA),\qquad y_i=\sqrt{m_i}M_i,\\
 s_i=\|M_i\|_2^2,\qquad w_i=\|y_i\|_2^2=m_i s_i,\qquad
 c_i=\int_{C_i}e(A)\,\mu(dA).
 \end{gathered}
\end{equation}
The amplitude integral exists by the integrability discussion in
Section~\ref{sec:selection}. The triangle inequality gives
$\|M_i\|_2\le m_i^{-1}\int_{C_i}\|A\|_2\,\mu(dA)=1$.
Thus $0\le s_i\le1$ and $0\le w_i\le m_i$.
Also $c_i<+\infty$, since $e\ge0$ and $c<+\infty$.

By Proposition A-2 in \cite[the Appendix]{shirokov2010} we can apply Jensen's inequality to the  convex, lower
semicontinuous and nonnegative function $\J$. It implies 
\begin{equation}\label{eq:jensen}
 \J(M_i)\le\frac1{m_i}\int_{C_i}\J(A)\,\mu(dA)
          =3+\frac{c_i}{m_i}<+\infty.
\end{equation}

By homogeneity and \eqref{eq:entropyidentity} we have
\begin{align}
 \E(y_i)&=m_i\E(M_i)
       =m_i\bigl[\J(M_i)-3s_i+s_i\log s_i\bigr]\notag\\
 &\le c_i+3m_i(1-s_i)+m_i s_i\log s_i\notag\\
 &\le c_i+3(m_i-w_i).
 \label{eq:meancost}
\end{align}
For $s_i=0$ the identity uses $M_i=y_i=0$ and $0\log0=0$.
For $s_i>0$ every term is finite by \eqref{eq:jensen}, and the
last inequality uses $s_i\le1$. In particular, each nonzero
extracted atom has finite normalized entanglement. We will now
identify $m_i-w_i$ as its cell's residual trace.

\medskip\noindent\textit{The covariance identity in trace class.}
Put
\[
 R_i=\int_{C_i}\proj{A}\,\mu(dA),\qquad
 b_i=\int_{C_i}A\,\mu(dA)=m_iM_i.
\]
On $C_i\times C_i$, the maps $\proj{A}$, $\proj{B}$,
$\ket{A}\bra{B}$ and $\ket{B}\bra{A}$ are strongly measurable
trace-class-valued functions of a pair $(A,B)$. Their trace norms are equal to $1$, while
$\|\proj{A-B}\|_1=\|A-B\|_2^2\le4$.
Since the product measure $\,\mu|_{C_i}\otimes\mu|_{C_i}\,$ has finite mass $m_i^2$, all are Bochner
integrable, and the Bochner Fubini theorem applies.

The two diagonal terms integrate to $m_iR_i$ each. The cross term
has the integral
\[
 \int_{C_i\times C_i}\ket{A}\bra{B}\,d\mu(A)d\mu(B)
    =\ket{b_i}\bra{b_i}=m_i^2\proj{M_i}.
\]
One may verify this by taking matrix elements against arbitrary
$x,z\in \mathcal{H}_{AB}$: the scalar integrand
$\langle x,A\rangle\langle B,z\rangle$ factors into a function
of $A$ and a function of $B$, and the Fubini theorem gives
$\langle x,b_i\rangle\langle b_i,z\rangle$.
The other cross term has the same integral. Thus expanding
\[
 \proj{A-B}=\proj{A}+\proj{B}
                 -\ket{A}\bra{B}-\ket{B}\bra{A}
\]
gives
\begin{equation}\label{eq:covariance}
 Q_i:=R_i-\proj{y_i}
   =\frac1{2m_i}\int_{C_i\times C_i}\proj{A-B}\,d\mu(A)d\mu(B)
   \ge0.
\end{equation}
The factor $1/(2m_i)$ cancels the two diagonal integrals and leaves
$\,m_i\proj{M_i}=\proj{y_i}\,$ from the cross terms. Positivity follows
also from testing the right side on $x\in \mathcal{H}_{AB}$: its quadratic form
is the integral of $|\langle A-B,x\rangle|^2/(2m_i)$.
In particular,
\begin{equation}\label{eq:cell-trace}
 \tr Q_i=m_i-w_i,\qquad
 \proj{y_i}\le R_i\le R.
\end{equation}

Because $\sum_i m_i=t$, the three positive series
$\sum_iR_i$, $\sum_i\proj{y_i}$ and $\sum_iQ_i$ converge
absolutely in trace norm: their trace-norm sums are respectively
$t$, $\sum_iw_i\le t$ and $\sum_i(m_i-w_i)\le t$.
Countable additivity of the Bochner integral gives $\sum_i R_i=R$.
Consequently
\begin{equation}\label{eq:remainder}
 R'=\sum_i Q_i
    =\sum_i\frac1{2m_i}\int_{C_i\times C_i}\proj{A-B}\,d\mu(A)d\mu(B),
 \qquad R=\sum_i\proj{y_i}+R'.
\end{equation}
This proves the exact operator decomposition and its positivity
before any normalization of the differences.

\medskip\noindent\textit{The residual as a Borel ensemble.}
For a same-cell pair write $D(A,B)=A-B$ and
$d(A,B)=\|A-B\|_2^2$. By Lemmas~\ref{lem:convex}
and~\ref{lem:cells},
\begin{equation}\label{eq:small}
 0\le d(A,B)\le\tfrac12\J(D)<\eta/2<1,
 \quad 0\le\E(D)\le\J(D)<\eta,\quad \textrm{where}\;\; D=D(A,B).
\end{equation}
Define a finite Borel measure on $\U\times\U$ by
\begin{equation}\label{eq:pair-measure}
 \nu=\sum_i\frac1{2m_i}
       (\mu|_{C_i})\otimes(\mu|_{C_i}),
 \qquad \nu(\U\times\U)=\sum_i\frac{m_i^2}{2m_i}=t/2.
\end{equation}
The nonnegative countable sum is a measure by Tonelli's theorem,
and the mass computation proves its finiteness. It is concentrated
on the Borel set $V=\bigcup_i(C_i\times C_i)$. Define
\[
 V_*=V\cap\{d(A,B)>0\},\qquad
 T(A,B)=\frac{A-B}{\sqrt{d(A,B)}},\quad(A,B)\in V_*.
\]
The maps $D,d$ are continuous on the set $\,\HS(\mathcal{H}_B,\mathcal{H}_A)\times\HS(\mathcal{H}_B,\mathcal{H}_A)$,
the map $T$ is continuous where $d(A,B)>0$. Thus, the map $T$ is Borel
on the set $V_*$. Its values have norm one. Because $d(A,B)<1$, the factor
$d(A,B)^{-1/2}$ is at least one on $V_*$, and \eqref{eq:scaling} gives
\[
 \J(T(A,B))\le d(A,B)^{-1}\J(D(A,B))<+\infty,\quad (A,B)\in V_*.
\]
Therefore $T$ takes values in $\U$, including the finite-entropy
condition needed at the next step.

Weight the measure $\nu|_{V_*}$ by the nonnegative Borel function $d$ and
push it forward by $T$:
\begin{equation}\label{eq:pushforward}
 \mu'=\bigl(d\,\nu|_{V_*}\bigr)\circ T^{-1},\qquad
 \mu'(F)=\int_{V_*}\mathbf1_F(T(A,B))\,d(A,B)\,\nu(dA,dB)
\end{equation}
for Borel $F\subset\U$. The weighted measure is finite by
\eqref{eq:small} and \eqref{eq:pair-measure}. For disjoint Borel
sets $F_j$, the monotone convergence theorem applied to their indicators gives
$\mu'(\bigcup_jF_j)=\sum_j\mu'(F_j)$, so the formula does define
a finite Borel measure. The set $\{D=0\}$ requires no choice of
normalized vector: its weight $d$ is zero.

For every nonnegative Borel function $g$ on $\U$, the pushforward
integration formula reads
\begin{equation}\label{eq:pushforward-test}
 \int_{\U}g(X)\,\mu'(dX)
   =\int_{V_*}d(A,B)g(T(A,B))\,\nu(dA,dB).
\end{equation}
It follows first for indicator functions from \eqref{eq:pushforward},
then for simple functions by addition, and finally for all such $g$
by the monotone convergence theorem.

To verify the barycenter, the integral
$\widehat R=\int_{\U}\proj{X}\,\mu'(dX)$ exists in trace class
because its integrand has norm one and $\mu'$ is finite. Apply
\eqref{eq:pushforward-test} to
$g(X)=|\langle X,z\rangle|^2$, for each $z\in \mathcal{H}_{AB}$. Since
$d\proj{T(A,B)}=\proj{D(A,B)}$ on $V_*$, this gives
\[
 \langle z,\widehat R z\rangle
  =\int_{V_*}|\langle D(A,B),z\rangle|^2\,\nu(dA,dB)
  =\langle z,R'z\rangle.
\]
The last equality uses \eqref{eq:remainder} and the fact that zero
differences contribute zero. Equality of quadratic forms, followed
by the polarization identity, proves $\widehat R=R'$.

\medskip\noindent\textit{Residual trace, residual cost and extracted cost.}
Putting $g=1$ in \eqref{eq:pushforward-test} yields
\begin{equation}\label{eq:trace-calculation}
 t'=\tr R'=\mu'(\U)=\int_{V_*}\,d(A,B)\nu(dA,dB)\le\frac{\eta}{2}\,\nu(V)=\frac{\eta t}{4}.
\end{equation}
For nonzero differences, homogeneity gives
$d(A,B)\,e(T(A,B))=\E(D(A,B))$. Since $e=\J-3$ is nonnegative
and Borel on $\U$, we may put $g=e$ in the same formula to obtain
\begin{equation}
 c'=\int_{\U}e(X)\,\mu'(dX)
    =\int_{V_*}\E(D(A,B))\,\nu(dA,dB)\leq\eta\nu(V)=\frac{\eta t}{2}<+\infty
 \label{eq:cost-calculation}
\end{equation}
due to (\ref{eq:small}). Although each normalized difference has finite entanglement, these
entanglements need not have a common upper bound. It is their
weighted integral, as displayed here, that is uniformly small.

Finally, taking traces in \eqref{eq:remainder} gives
$t'=t-\sum_iw_i=\sum_i(m_i-w_i)$. Summing
\eqref{eq:meancost} over cells, with nonnegative terms, proves
\[
 \sum_i\E(y_i)\le\sum_i c_i+3\sum_i(m_i-w_i)=c+3t'.
\]
Together with \eqref{eq:trace-calculation} and
\eqref{eq:cost-calculation}, this proves all the assertions.
\end{proof}

\subsection{Iteration at the original barycenter}\label{sec:iteration}
\begin{proof}[Final step of the proof of Theorem~\ref{thm:discretization}]
Use the Borel section of Section~\ref{sec:selection} to lift the
original measure $\mu$ to a probability measure $\mu_0$ on $\U$. Its
barycenter is $R_0=\rho$ and its cost is $c_0=c_\mu$. Fix
$0<\eta<1$ and put $\alpha=\eta/4$.

\medskip\noindent\textit{Recursive measures and estimates.}
Suppose $\mu_n$ has already been defined as a finite Borel measure
on $\U$, with finite entanglement cost. Set
\[
 R_n=\int_{\U}\proj{A}\,\mu_n(dA),\qquad
 t_n=\mu_n(\U)=\tr R_n,\qquad
 c_n=\int_{\U}e(A)\,\mu_n(dA).
\]
Apply Lemma~\ref{lem:extraction} with this measure and the same
$\eta$. Denote its output atoms by $(y_{n,i})_i$ and its weighted
pushforward residual measure by $\mu_{n+1}$. Thus
\begin{align}
 R_n&=P_n+R_{n+1},\qquad P_n=\sum_i\proj{y_{n,i}}\ge0,
                  \quad R_{n+1}\ge0,\label{eq:recursive-operators}\\
 t_{n+1}&\le\alpha t_n,\qquad
 c_{n+1}\le\frac{\eta}{2}t_n,\label{eq:recursion}\\
 a_n&:=\sum_i\E(y_{n,i})\le c_n+3t_{n+1}.
                  \label{eq:stage-cost}
\end{align}
The lemma proves that $\mu_{n+1}$ is again a finite Borel measure
on $\U$ and that $c_{n+1}<+\infty$. These are precisely the
hypotheses for the next application; we do not replace the residual
by an arbitrary pure-state decomposition. If $t_n=0$, then
$\mu_n=0$ and $R_n=0$; take all subsequent measures, atoms and
costs to be zero. This also covers finite termination.

Induction gives, for all $n\ge0$,
\begin{equation}\label{eq:geometric}
 t_n\le\alpha^n,\qquad 0\le R_{n+1}\le R_n\le\rho,
 \qquad
 c_n\le\frac{\eta}{2}\alpha^{n-1}\quad(n\ge1).
\end{equation}
The first inequality starts from $t_0=1$ and uses
\eqref{eq:recursion}. The operator inequalities follow from
\eqref{eq:recursive-operators}, and the last bound follows from
$c_n\le(\eta/2)t_{n-1}$. Notice that the initial cost $c_0$ is
kept separately. Every subsequent cost is bounded by the preceding
residual mass.

\medskip\noindent\textit{An exact countable barycenter.}
Summing \eqref{eq:recursive-operators} for $n=0,\ldots,N-1$
gives the exact identity
\begin{equation}\label{eq:partial-barycenter}
 \rho=\sum_{n=0}^{N-1}P_n+R_N,
 \qquad
 \left\|\rho-\sum_{n=0}^{N-1}P_n\right\|_1
       =\|R_N\|_1=t_N\le\alpha^N.
\end{equation}
The equality of trace and trace norm uses $R_N\ge0$. Because $0<\alpha<1$,
the remainder tends to zero in trace norm. This is a geometric
estimate after $N$ extraction stages. A single stage can already
extract countably many atoms, so \eqref{eq:partial-barycenter} does
not bound the remainder after a prescribed finite number of ensemble
members.

For each stage $n$ the atom series defining $P_n$ is trace-norm
convergent. Its total trace is $t_n-t_{n+1}$, so
\[
 \sum_{n\ge0}\sum_i\|\proj{y_{n,i}}\|_1
  =\sum_{n\ge0}(t_n-t_{n+1})
  =t_0-\lim_{N\to+\infty}t_N=1.
\]
Thus the series over all pairs $(n,i)$ is absolutely summable in
the Banach space of trace-class operators. The set of pairs is
countable, and any enumeration of these pairs has the same sum.
By \eqref{eq:partial-barycenter}, that sum is $\rho$.

Omit zero atoms. For every remaining pair $(n,i)$ put
\[
 q_{n,i}=\|y_{n,i}\|_2^2>0,\qquad
 \psi_{n,i}=y_{n,i}/\sqrt{q_{n,i}}.
\]
The amplitudes $\psi_{n,i}$ represent unit vectors in $\mathcal{H}_{AB}$. The
preceding trace calculation proves
\begin{equation}\label{eq:final-ensemble}
 \sum_{n,i}q_{n,i}=1,\qquad
 \rho=\sum_{n,i}q_{n,i}\proj{\psi_{n,i}}
       \quad\text{in trace norm}.
\end{equation}
Repeated pure states may be retained as separate labels or
combined by adding their weights. The range remark after this proof
verifies the range condition for each atom. No requirement that an output vector
belong to the support of the original measure $\mu$ has entered the argument.

\medskip\noindent\textit{Summation of the entanglement costs.}
Homogeneity gives $q_{n,i}e(\psi_{n,i})=\E(y_{n,i})$.
All summands are nonnegative, so the sum can be grouped by stage.
Using \eqref{eq:stage-cost} and then \eqref{eq:recursion} yields
\begin{align}
 \sum_{n,i}q_{n,i}e(\psi_{n,i})
 &=\sum_{n\ge0}a_n\notag\\
 &\le c_0+\sum_{n\ge1}c_n+3\sum_{n\ge0}t_{n+1}\notag\\
 &\le c_\mu+\frac{\eta}{2}\sum_{n\ge0}t_n
                +3\alpha\sum_{n\ge0}t_n\notag\\
 &\le c_\mu+\frac{\eta/2+3\eta/4}{1-\eta/4}
       =c_\mu+\frac{5\eta}{4-\eta}.
 \label{eq:total-cost}
\end{align}
The geometric bound on $t_n$ justifies finiteness of both series.
Equations~\eqref{eq:final-ensemble} and~\eqref{eq:total-cost}
prove the quantitative discretization assertion.
\end{proof}

\begin{unnumberedremark}[Location of the extracted atoms]
The operator inequality in \eqref{eq:cell-trace} also verifies the
range condition for an atom in an exact decomposition. Suppose the
incoming $R$ satisfies $R\le\rho$. Then
$\proj{y_i}\le R_n\le R\le\rho$ for a given vector $y_i$ extracted at the $n$-th iteration.
Write the spectral decomposition on the support of $\rho$ as
$\rho=\sum_j\lambda_j\proj{u_j}$, where $\lambda_j>0$.
For $z\in\ker\rho$, domination gives
$|\langle y_i,z\rangle|^2\le\langle z,\rho z\rangle=0$.
Thus $y_i$ lies in the closed span of $(u_j)$.

Put $a_j=\langle u_j,y_i\rangle$ and, for a finite initial set,
$z_N=\sum_{j\le N}a_j\lambda_j^{-1}u_j$. The same domination gives
\[
 \left(\sum_{j\le N}\frac{|a_j|^2}{\lambda_j}\right)^2
   =|\langle y_i,z_N\rangle|^2
   \le\langle z_N,\rho z_N\rangle
   =\sum_{j\le N}\frac{|a_j|^2}{\lambda_j}.
\]
Each nonnegative partial sum is therefore at most one. The series
$x=\sum_j a_j\lambda_j^{-1/2}u_j$ converges in $\mathcal{H}_{AB}$ with
$\|x\|\le1$ and satisfies $\rho^{1/2}x=y_i$. Hence
\begin{equation}\label{eq:range}
 y_i\in\operatorname{Ran}\rho^{1/2},\qquad
 \|\rho^{-1/2}y_i\|\le1.
\end{equation}
Here the inverse is taken on the support of $\rho$. This is
membership in the actual range, not just its closure (equal to the support of $\rho$). In normalized
notation $y_i=\sqrt{q_i}\psi_i$, it says
$q_i\|\rho^{-1/2}\psi_i\|^2\le1$. The proof also applies to
finite-rank $\rho$ by using its finite spectral sum. The iteration
has $R_n\le\rho$ at every stage, so \eqref{eq:range} holds
for every output atom.
\end{unnumberedremark}

The only entropy integrability assumption in this construction is
$\int e\,d\mu<+\infty$. In particular, no step uses finiteness of
$S(\rho_A)$, $S(\rho_B)$, or $S(\rho)$ (cf. \cite{shirokov2010}).

The use of $\J$ is specific to the entanglement cost. Its convexity
controls the means, and its small-difference bound controls the
weighted residual entanglement. A general lower semicontinuous cost function
need not have these properties: \cite[Remark~6]{shirokov2010}
gives distinct discrete and continuous convex roofs for some  bounded
lower semicontinuous unitarily unvariant concave cost function.

\begin{corollary}\label{cor:consequences}
The discrete Entanglement of Formation (the function defined in (\ref{eq:discrete-eof})) is lower semicontinuous on
the full state space. It vanishes on every separable state. For a
countably nondecomposable separable state the  infimum in (\ref{eq:discrete-eof}) is not attained.\footnote{A separable
state is countably nondecomposable if it has no finite or countable
decomposition into pure product states \cite{hsw2005}.}
\end{corollary}
\begin{proof}
The first assertion follows from equality with the continuous  EoF (the function defined in (\ref{eq:continuous-eof})) \cite[Section~5.1]{shirokov2010}. Every separable state
has a Borel measure representation by pure product states
\cite[Lemma~1]{hsw2005}, and therefore has zero continuous  EoF.
An ensemble attaining zero discrete EoF would consist entirely of
pure product states, since all summands are nonnegative. This is
excluded by countable nondecomposability.
\end{proof}

\paragraph{Use of AI tools.}
In preparing this work, Haobo Ma and Wenlin Zhang used an AI
orchestration built on OpenAI's ChatGPT (including GPT Pro) and Codex
and Anthropic's Claude. The AI agents assisted with literature search,
symbolic and numerical computations, drafting and checking proof details,
and Lean formalization attempts. Separate ChatGPT and Codex runs were
used to cross-check the arguments, and Claude coordinated the work.
All mathematical content was checked by the authors, who take full
responsibility for it.

\end{document}